\documentclass[journal]{IEEEtran}
\usepackage{dsfont}
\newcommand{\ind}{\mathds{1}}
\usepackage{cite}
\usepackage{amsmath,amssymb,amsfonts}
\usepackage{algorithm}
\usepackage{algpseudocode}
\usepackage{graphicx}
\usepackage{booktabs}
\usepackage{array}
\usepackage{multirow}
\usepackage{xcolor}
\usepackage{url}
\usepackage[hidelinks]{hyperref}

\graphicspath{{figures/}}
\newtheorem{proposition}{Proposition}

\newcommand{\TC}{\textrm{TC}}
\newcommand{\NTC}{\textrm{NTC}}
\newcommand{\wtc}{w_{\TC}}
\newcommand{\wntc}{w_{\NTC}}
\newcommand{\tcr}{\mathrm{TCR}_{\TC}}

\begin{document}

\title{When Do LLM Agents Help? Deadline-Aware Mixed-Criticality Task Scheduling at the Autonomous-Vehicle Edge}

\author{Reza~Zakerian\\
Westcliff University\\
r.zakerian.8143@westcliff.edu}

\markboth{}%
{Zakerian: When Do LLM Agents Help? Deadline-Aware Mixed-Criticality Task Scheduling at the Autonomous-Vehicle Edge}
\maketitle

\begin{abstract}
Autonomous vehicles offload latency-sensitive perception tasks to nearby mobile
edge computing (MEC) servers, where a missed safety-critical task is unsafe
rather than merely degraded. Large language models (LLMs) are increasingly
proposed as adaptive, explainable schedulers, yet evidence of when they help is
scarce. We study deadline-aware, mixed-criticality scheduling on heterogeneous
MEC servers, where time-critical (TC) tasks must be protected at a controlled
cost to best-effort traffic, and ask whether a multi-agent LLM control layer
improves on a strong heuristic. We answer in two steps. First we build the
heuristic: a windowed contract-net auction that orders each admission window
time-critical-first by earliest deadline and places tasks by
earliest-finish-time. Across $60$ instances on three topologies and $15$
baselines under an identical online constraint, it attains a TC completion rate
of $0.902$, above every baseline (Holm-corrected $p<0.001$; best baseline
$0.838$) and at $0.87$ of a CP-SAT upper bound. Second, we add the LLM control
plane. A controlled decomposition traces the scheduler's advantage to two
ordinary factors, the batching horizon and time-critical-first ordering; the
auction, the per-window LLM policy, and online adaptation add nothing while the
load is stationary, where the heuristic is already near-optimal. Under a mid-run surge of safety-critical tasks the picture
changes, and the LLM control plane gains significantly over both the static
heuristic and the bandit. LLM orchestration therefore earns its cost only when
non-stationarity opens headroom a fixed policy cannot use. We report
control-plane latency and rationale, and release all code and
seeded instances.\footnote{Code at:
\url{https://github.com/rezarz98/multi-agents-llm-edge-scheduling}.}
\end{abstract}


\section{Introduction}
\IEEEPARstart{A}{utonomous} vehicles perceive their environment by continuously
processing data from cameras, LiDAR, and radar, and much of this computation is
offloaded to nearby mobile edge computing (MEC) servers through roadside units
(RSUs) to meet stringent timing constraints~\cite{sars2024,mao2017survey}. The
edge scheduler that assigns these offloaded tasks to heterogeneous servers must
operate online and must complete as many tasks as possible before their
deadlines. The setting is inherently \emph{mixed-criticality}: a missed
time-critical (TC) task, such as obstacle detection feeding the control loop, can
be unsafe, whereas a missed best-effort non-time-critical (NTC) task, such as an
infotainment or map-update request, is merely degraded~\cite{vestal2007preemptive,burns2017survey}.
The scheduler must therefore protect time-critical work while limiting the
collateral cost to best-effort traffic.

Existing approaches each solve part of this problem. Classical dispatching rules
are online and deadline-aware, but myopic and blind to
placement~\cite{liu1973scheduling,dertouzos1989online}. Mapping heuristics choose
servers well, yet optimise makespan rather than
deadlines~\cite{ibarra1977heuristic,topcuoglu2002heft}. Metaheuristic and
deep-reinforcement-learning (DRL) schedulers optimise a fixed instance or train
offline, so they adapt poorly once conditions
change~\cite{wang2022ipso,zeng2024ddqn}.

Large language models (LLMs) have emerged as general reasoners and the control
core of multi-agent systems~\cite{brown2020language,wang2024agentsurvey}, and a
fast-growing literature applies them to scheduling and resource
management~\cite{jadhav2025llmhpc}. Their appeal is adaptivity with
natural-language interpretability. Their obstacle is cost: LLM inference is
orders of magnitude slower than a heuristic dispatch, so a model cannot sit on
the per-task critical path. This leaves an architectural question that, to our
knowledge, has not been answered with rigorous evidence for autonomous-vehicle
edge scheduling: can a lightweight LLM control layer, kept off the critical path,
improve on a competitive deadline-aware scheduler, and if so, when?

We answer in two steps, and the answer turns out to be conditional. First we
build the scheduler, a \emph{windowed contract-net} auction in which a broker
announces release-ordered admission windows, per-server agents bid, and each task
is placed by earliest-finish-time, with time-critical priority enforced through
ordering rather than admission refusal. Second, we add a lightweight LLM control
plane of broker, edge, and monitor agents that set policy once per window, and
test it against strong non-LLM controls. The scheduler proves near-optimal while
the workload is stationary, which leaves an adaptive layer nothing to recover;
only when a mid-run surge opens headroom does the LLM control plane help.

The evaluation is built to make that answer trustworthy. All $15$ baselines run
under the same online constraint. We generate $60$ independent instances and
report paired significance tests with effect sizes. We compute the CP-SAT offline
optimum to bound what any scheduler could achieve. Finally, we decompose our own
method, removing one component at a time, to locate the source of its advantage.

Our contributions are as follows:
\begin{itemize}
\item \emph{A simple, optimum-anchored scheduler} for deadline-aware,
mixed-criticality autonomous-vehicle task offloading. It admits tasks in windows,
protects time-critical work by ordering rather than by refusing best-effort
traffic, and places each task on the server that can finish it soonest, all
realised as a contract-net auction. Across three topologies and $60$ instances it
outperforms $15$ standard baselines and reaches $0.87$ of a CP-SAT upper bound,
an optimality gap of at most $13\%$ (Sec.~\ref{sec:overall}).

\item \emph{A controlled decomposition} that locates the advantage. Two ordinary
factors are responsible, the batching horizon and time-critical-first ordering,
rather than the auction or any learned component. A windowed time-critical-first
heuristic reproduces the full method within confidence intervals, and neither an
LLM control plane nor online adaptation, by an LLM monitor or a UCB1 bandit,
improves on it while the load is stationary (Sec.~\ref{sec:decomp}).

\item \emph{An evidence-based answer to when LLM orchestration helps.} Under a
non-stationary surge of safety-critical tasks, where the static heuristic is no
longer near-optimal, the LLM control plane gains over both that heuristic and a
non-LLM bandit ($+0.005$, $p=0.004$). It gives no benefit under stationary load
or capacity loss (Sec.~\ref{sec:disruption}). We pair this with measurements of
control-plane latency, decision frequency, and rationale
(Sec.~\ref{sec:llm}).
\end{itemize}

The rest of the paper is organised as follows. Sec.~\ref{sec:related} reviews
related work; Sec.~\ref{sec:model} presents the system model and problem
formulation; Sec.~\ref{sec:method} describes the proposed scheduler and its two
planes; Sec.~\ref{sec:setup} details the experimental setup; and
Sec.~\ref{sec:results} reports the results. Sec.~\ref{sec:conclusion} concludes.

\section{Related Work}\label{sec:related}
\textbf{Dispatching rules and real-time scheduling.} First-come-first-served
(FCFS), earliest-deadline-first (EDF)~\cite{liu1973scheduling}, earliest-due-date
(EDD), earliest-feasible-deadline-first (EFDF), critical-ratio (CR),
cost-over-time (COVERT), and efficient-resource-allocation (ERA)
rules~\cite{pinedo2016scheduling,sars2024}, together with least-laxity-first
(LLF)~\cite{dertouzos1989online}, order pending tasks from local features. They
are cheap and naturally online, but they look only one task ahead and say nothing
about which server should run it.

\textbf{Heterogeneous mapping.} Min-Min, Max-Min~\cite{ibarra1977heuristic,braun2001comparison},
Sufferage~\cite{maheswaran1999dynamic}, and HEFT~\cite{topcuoglu2002heft} map
independent tasks to machines to minimise makespan. Maheswaran \emph{et
al.}~\cite{maheswaran1999dynamic} distinguish immediate and batch dynamic modes
for online use, a distinction we adopt so that every baseline is treated fairly.
These methods place tasks well, but they are deadline-agnostic.

\textbf{Learning-based scheduling.} Metaheuristics such as PSO~\cite{kennedy1995pso,wang2022ipso}
and GA~\cite{goldberg1989genetic,wang2024hybrid}, and DRL from the Deep
Q-Network~\cite{mnih2015human} and its double variant~\cite{vanhasselt2016double}
to edge makespan minimisation~\cite{zeng2024ddqn} and MEC
offloading~\cite{liu2021online}, learn capable policies. Both families pay for
that strength: they optimise a fixed instance or train offline, which makes them
slow to adapt and hard to interpret online.

\textbf{Mixed-criticality scheduling.} Vestal's model~\cite{vestal2007preemptive}
and the survey of Burns and Davis~\cite{burns2017survey} formalise systems that
mix criticality levels; our TC/NTC weighting is a simple two-level instance of
this setting, adapted to throughput maximisation for vehicular offloading.

\textbf{LLMs for systems.} LLMs are capable reasoners~\cite{brown2020language}
and multi-agent controllers~\cite{wang2024agentsurvey}, and a recent literature
applies them to scheduling and resource management~\cite{jadhav2025llmhpc}. Cost
rules out putting them in a data-plane role. We differ from this work in two
ways: we confine LLM reasoning to a periodic control plane, and we measure its
contribution against non-LLM controls and an offline optimum rather than
reporting only the cases where it succeeds. The result is a conditional claim
instead of an absolute one.

No single family above covers the setting we need. We therefore begin from a
model that makes deadlines, criticality, and placement explicit at once.

\section{System Model and Problem Formulation}\label{sec:model}
We extend the two-layer hierarchical autonomous-vehicle MEC architecture of
Zakerian and Gholami~\cite{sars2024} with a control layer, and add
task-criticality classes and a weighted-throughput objective; the resulting
three-layer design (vehicles, RSUs, edge servers, control plane) is shown in
Fig.~\ref{fig:arch}, and Table~\ref{tab:notation} summarises the notation.

\begin{table}[t]
\centering
\caption{Summary of notation.}
\label{tab:notation}
\renewcommand{\arraystretch}{1.12}
\begin{tabular}{@{}ll@{}}
\toprule
Symbol & Meaning \\
\midrule
$\mathcal{T}=\{t_1,\dots,t_n\}$ & set of $n$ independent tasks \\
$t_i^r,\;t_i^d$ & release time and deadline of $t_i$ \\
$t_i^w,\;t_i^s$ & workload (MI) and file size (KB) of $t_i$ \\
$\kappa_i\in\{\TC,\NTC\}$ & criticality class of $t_i$ \\
$\mathcal{E}=\{ec_1,\dots,ec_m\}$ & heterogeneous edge computing servers (ECS) \\
$c_k^j,\;r_k^j,\;\tau_k^j$ & PU $j$ of $ec_k$; its rate; its free time \\
$\delta_i^{\mathrm{ap}},\delta_i^{\mathrm{br}},\delta_{i,k}^{\mathrm{bs}}$ &
vehicle--RSU, broker, base-station delays \\
$C_i(c_k^j)$ & completion time of $t_i$ on PU $c_k^j$ \\
$x_i\in\{0,1\}$ & $1$ if $t_i$ completes before its deadline \\
$\wtc,\wntc$ & objective weights (TC, NTC) \\
$W$ & auction (admission) window size \\
$R$ & monitor review period \\
\bottomrule
\end{tabular}
\end{table}

\subsection{Task and Delay Model}
Vehicles generate tasks $t_i=\langle t_i^r,t_i^d,t_i^w,t_i^s,\kappa_i\rangle$
with release $t_i^r$, deadline $t_i^d>t_i^r$, workload $t_i^w$ (MI), file size
$t_i^s$, and class $\kappa_i$; each must complete within $[t_i^r,t_i^d]$. Servers
$ec_k$ host non-preemptive processing units (PUs) $c_k^j$ of rate $r_k^j$, available at $\tau_k^j$.
The computation delay is $t_i^w/r_k^j$. Communication delay is incurred in
transit, overlapping any queueing at the PU: $t_i$ arrives at server $ec_k$ at
time $t_i^r+\delta_i^{\mathrm{ap}}+\delta_i^{\mathrm{br}}+\delta_{i,k}^{\mathrm{bs}}$,
summing the vehicle-to-RSU access hop $\delta_i^{\mathrm{ap}}$, the RSU-to-broker
relay $\delta_i^{\mathrm{br}}$, and the broker-to-server link
$\delta_{i,k}^{\mathrm{bs}}$; it then contends for the PU, which it occupies for
its processing time alone. Hence
\begin{equation}
C_i(c_k^j)=\max\!\Big(\tau_k^j,\;t_i^r+\delta_i^{\mathrm{ap}}
+\delta_i^{\mathrm{br}}+\delta_{i,k}^{\mathrm{bs}}\Big)+\frac{t_i^w}{r_k^j}.
\label{eq:ct}
\end{equation}

This corrects a common additive form that adds all communication after the PU
becomes free and thereby double-counts delay under contention. A task completes
if and only if $C_i(c_k^j)\le t_i^d$ on its assigned PU, that is,
$x_i=\ind[C_i(c_k^j)\le t_i^d]$. A task that cannot meet its deadline is
dropped and does not occupy a PU.

\subsection{Objective and Complexity}
With weights $\wtc>\wntc>0$, the scheduler chooses online an ordering and an
assignment to maximise weighted on-time throughput,
\begin{equation}
\max \sum_{i=1}^{n} w(\kappa_i)\,x_i,
\label{eq:obj}
\end{equation}

subject to each PU serving one task at a time, non-preemptively, using at any
time $t$ only tasks with $t_i^r\le t$. The primary measure is the time-critical
completion rate $\tcr=\big(\sum_{i:\kappa_i=\TC}x_i\big)\big/\,|\{i:\kappa_i=\TC\}|$,
with \eqref{eq:obj} as secondary. This generalizes the maximum-processed-tasks
objective of~\cite{sars2024} to the mixed-criticality setting.

\begin{figure}[t]
\centering
\includegraphics[width=\columnwidth]{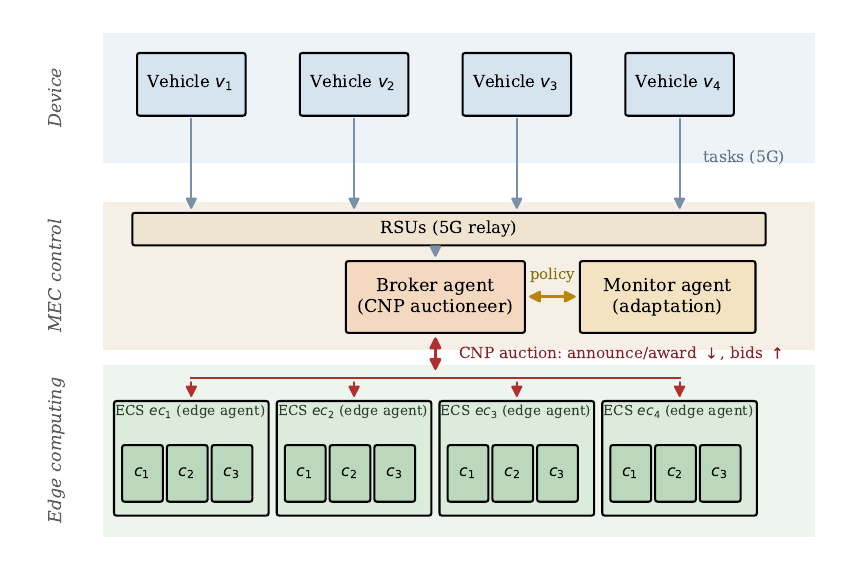}
\caption{Three-layer autonomous-vehicle MEC architecture with broker, edge, and
monitor agents.}
\label{fig:arch}
\end{figure}

\begin{proposition}
Maximising weighted on-time throughput \eqref{eq:obj} is NP-hard.
\end{proposition}
\noindent\emph{Proof sketch.} Even the single-machine special case with release
dates, minimizing the weighted number of late jobs
($1\,|\,r_j\,|\,\sum w_j U_j$), equivalently maximizing the weighted number of
on-time jobs, is strongly NP-hard~\cite{lenstra1977complexity}; our
unrelated-machine, communication-delay setting is a generalization. 
\smallskip

\noindent This hardness motivates heuristic and, potentially, learned
schedulers, and frames the offline optimum we compute in Sec.~\ref{sec:setup} as
a bound rather than an operational target.

\subsection{Modelling Assumptions}
Three idealisations are worth stating. First, communication
delays are deterministic functions of file size and path; a stochastic-channel
study is future work. Second, dropped tasks are removed at zero cost, which
grants every scheduler perfect infeasibility foresight and mildly advantages
methods that explicitly test feasibility, including ours. Third, there is no
queueing at the RSU or broker. These are standard for this class of study but
should temper absolute-rate interpretation; our comparisons hold all of them
fixed across methods.

\section{Proposed Scheduler}\label{sec:method}
We build the scheduler in four steps: why the cost of LLM inference dictates a
two-plane design; the deterministic \emph{data plane} that dispatches every task;
the LLM \emph{control plane} that configures it; and the full procedure with its
cost and its fallback.

\begin{figure}[t]
\centering
\includegraphics[width=\columnwidth]{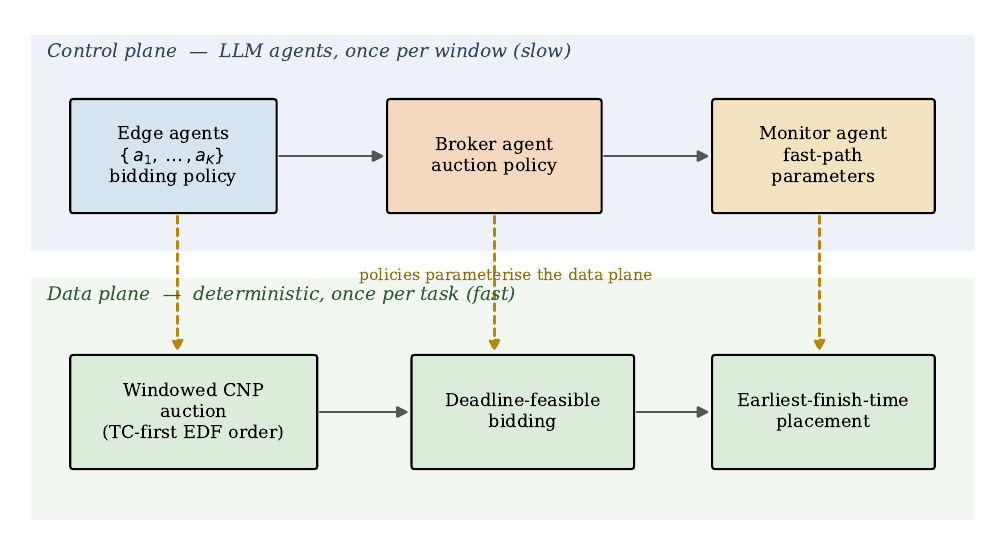}
\caption{Two-plane design: LLM agents set policy, deterministic auction
dispatches tasks.}

\label{fig:framework}
\end{figure}

\subsection{Design Rationale}
The cost of LLM inference, rather than any property of the scheduling problem,
dictates the architecture. A model call takes seconds, whereas an admission
decision must be taken in the time it takes a PU to free; placing a
language model on the per-task path is therefore not a tuning choice but an
impossibility. Any useful role for an LLM here must be one in which it is
consulted rarely and its output is reused many times.
 
This leads to a two-plane separation (Fig.~\ref{fig:framework}). A deterministic
\emph{data plane} dispatches every task using fixed arithmetic on PU clocks. A
lightweight \emph{control plane} of LLM agents sets the parameters that the data
plane obeys, once per admission window. Because the control plane is invoked per
window rather than per task, the number of model calls scales with $n/W$ and is
independent of the task arrival rate; the data plane's latency is unaffected by
the presence or absence of the LLM. The separation also makes the empirical
question of this paper well posed: the two planes can be enabled and disabled
independently, so the contribution of LLM reasoning can be measured against the
same data plane running under a fixed policy.
 
\subsection{Data Plane: Windowed Contract-Net}
 
\subsubsection{Windows and decision epochs}
Released tasks are grouped into \emph{admission windows} of $W$ tasks in release
order, and the scheduler handles one window at a time. A window $\omega$ is
complete only once its $W$th task has been released. We therefore define its
\emph{decision epoch} $t^{\mathrm{cl}}_\omega$ as the release time of that last
task, the earliest instant at which the window's contents are known.

This horizon carries most of the scheduler's advantage
(Sec.~\ref{sec:decomp}), so it is worth being precise about what it does and does
not provide. It grants an \emph{intra-window ordering lookahead}: the scheduler
sees all $W$ tasks before committing to any of them, and can order them by
criticality and deadline instead of reacting to one arrival at a time. It grants
no foreknowledge of unreleased work. Tasks enter windows strictly in release
order, and by the $\max$ in \eqref{eq:ct} no task begins before it physically
arrives at its server.

The window size $W$ thus trades accuracy against delay. A larger window orders
more tasks jointly and raises time-critical completion, but it also holds every
decision until $t^{\mathrm{cl}}_\omega$, a delay that grows with $W$ at a fixed
arrival rate. We use $W{=}40$ throughout and report the full sweep in
Sec.~\ref{sec:decomp}.
 
\subsubsection{Bidding}
Within a window, scheduling follows the contract-net
protocol~\cite{smith1980contract}. The broker announces the window's tasks to the
$m$ edge agents, one per server. For each announced task $t_i$, the agent of
server $ec_k$ evaluates \eqref{eq:ct} across its own PUs and replies with a
single bid, the earliest completion time it can offer,
\begin{equation}
b_{i,k}=\min\nolimits_{j} C_i(c_k^j),
\label{eq:bid}
\end{equation}
together with a feasibility flag $\ind[\,b_{i,k}\le t_i^d\,]$ marking whether
that bid meets the task's deadline. Bids are computed from locally held state
only, namely the agent's own PU rates and free times, so no global schedule needs
to be materialised.
 
\subsubsection{Ordering and award}
The broker sorts the window \emph{time-critical-first}, and by earliest deadline
within each criticality class, then awards tasks in that order. Each task is
awarded to the agent offering the earliest \emph{feasible} bid, ties broken by
current load; the winning agent executes the task on the PU that produced its
bid, and updates that PU's free time $\tau_k^j$ before the next award. A task
with no feasible bid is dropped and occupies no PU.
 
We call this rule \emph{windowed TC-EDF+EFT}: earliest-deadline-first ordering
applied time-critical-first (TC-EDF), with earliest-finish-time (EFT) placement,
over windows of size $W$. Two properties matter for the analysis that follows.

First, criticality is enforced purely by \emph{ordering}, never by refusing
best-effort work. An NTC task is delayed behind TC tasks but still bids and can
still win, so critical traffic is protected by contention rather than by
admission policy. Second, every step of the loop is arithmetic on PU clocks, so
the data plane is deterministic and no LLM sits on this path. Greedy EFT
placement is not the only reasonable choice, and Sec.~\ref{sec:decomp} quantifies
what it costs relative to spreading load across all deadline-feasible servers.
 
\subsection{Control Plane: LLM Agents}
Once per window, three agent types set the fast-path parameters $\Theta$ that the
data plane will obey. Each agent receives a compact JSON description of the
current state: per-server load and free-time summaries, rolling completion and
deadline-miss counts by class, and the current value of $\Theta$. Each replies
through a constrained output schema, so a reply is either a valid parameter
assignment or a detectable failure, never free text the data plane must
interpret.

The three roles differ in scope. \emph{Edge agents} set a bidding policy from
local load. \emph{The broker agent} sets the tie-break rule and decides whether
to drop best-effort tasks that provably cannot finish in time, freeing capacity
earlier than waiting for them to fail. \emph{The monitor agent} works on a slower
cadence: every $R$ windows it reviews the rolling statistics and adapts $\Theta$,
namely the ordering policy (strict EDF versus a randomised top-$k$ EDF) and the
degree of best-effort deferral. Every monitor decision is logged with the
natural-language rationale the model gave for it, which we examine in
Sec.~\ref{sec:llm}.

The monitor's authority is deliberately small. With $W{=}40$ and $R{=}2$, a
$200$-task instance yields five windows and two review points, so the monitor
acts at most twice. This bounds how much it can contribute, and we return to the
point when interpreting its null result.

Adaptation and LLM reasoning are separable, so we separate them. As a non-LLM
control we implement a \emph{UCB1 bandit} over the identical parameter set, with
the same review cadence and the same authority over $\Theta$. Any advantage the
LLM shows over this control therefore cannot be credited to the mere presence of
an adaptive loop, which is what makes the comparison in
Sec.~\ref{sec:disruption} informative rather than merely favourable.
 
\subsection{Procedure, Cost, and Fallback}
Algorithm~\ref{alg:has} states the full procedure. The control plane occupies
lines 4--7 and runs once per window; the data plane occupies lines 8--12, which
order the window under the current $\Theta$ and then dispatch its tasks
deterministically. Setting the optional lines 5 and 6 aside recovers the pure
heuristic scheduler used as the static control throughout Sec.~\ref{sec:results}.
 
\begin{algorithm}[t]
\caption{Windowed contract-net scheduler (with optional control plane)}
\label{alg:has}
\begin{algorithmic}[1]
\Require tasks $\mathcal{T}$, servers $\mathcal{E}$, window size $W$, review period $R$
\State $\Theta\gets$ default fast-path parameters
\State partition $\mathcal{T}$ into release-ordered windows of size $W$
\For{each window $\omega$ with index $q$}
  \State broker announces $\omega$; edge agents report load
  \State \textbf{(optional)} edge/broker agents set policy \Comment{LLM}
  \If{$q\bmod R=0$} adapt $\Theta$ via monitor (LLM) or bandit \EndIf
  \State order $\omega$ time-critical-first, EDF within class, using $\Theta$
  \For{each task $t_i$ in the ordered window}
     \State collect bids $\{b_{i,k}\}$; award earliest feasible bid
     \State execute $t_i$; update $\tau_k^j$; record $x_i$
  \EndFor
\EndFor
\State \Return per-class completion statistics
\end{algorithmic}
\end{algorithm}
 
\emph{Data-plane cost.} For $n$ tasks, $m$ servers of $P$ PUs each, every
announced task draws one bid per server and each bid minimises over that server's
PUs, giving $O(nmP)$ finish-time evaluations in total. Ordering contributes
$O(n\log W)$ across all windows, which is dominated by the bidding term for any
realistic $m$ and $P$. The cost is independent of $\Theta$, so enabling the
control plane does not change data-plane complexity.
 
\emph{Control-plane cost.} Each window issues one call per edge agent and one to
the broker, and the monitor is consulted every $R$ windows, so the number of
model calls per instance is
\begin{equation}
\left\lceil \frac{n}{W} \right\rceil (m+1)
\;+\;\left\lfloor \frac{n}{WR} \right\rfloor
\;=\;O\!\left(\frac{nm}{W}\right),
\label{eq:llmcost}
\end{equation}
which for the configuration used here ($n{=}200$, $W{=}40$, $m{=}4$, $R{=}2$)
gives $27$ calls, matching the live measurement in Sec.~\ref{sec:llm}. The
critical property is the $1/W$ factor: model cost falls as the horizon grows,
while data-plane cost does not, so the two planes can be sized independently.
 
\emph{Fallback.} Because control-plane latency is not bounded by the window
duration, a deployment cannot assume a policy update will arrive before the
window it was requested for. The data plane therefore always holds a committed
$\Theta$ and proceeds with it if no reply has arrived, falling back to the last
committed value or to the heuristic default. Sec.~\ref{sec:llm} measures the
latency that makes this necessary.

\section{Experimental Setup}\label{sec:setup}
\textbf{Instances.} Each instance has $200$ tasks with a $60/40$ TC/NTC split.
Release times, deadlines, workloads, and file sizes are resampled independently
from the empirical marginals of a contended reference benchmark, giving fresh
realisations of the same statistical character. We evaluate $20$ seeds on each of
three MEC topologies (two four-server, one six-server), $60$ instances in total;
the same task core is used across topologies so that only the network differs.

The load is contended but not saturated, which is what makes the benchmark
discriminative. Averaged over instances, $98.5\%$ of TC tasks are \emph{feasible
in isolation}, meaning some PU could complete them in an otherwise-empty system.
Only about $1.5\%$ are therefore impossible by construction, and the rest of the
shortfall below that ceiling comes from contention, which is exactly what a
scheduler can act on.

\textbf{Baselines and controls.} We compare against $15$ baselines. All of them
run through one shared \emph{event-driven online} dispatch engine that exposes
only released tasks: a decision is taken whenever a server frees, over whatever
has arrived by that instant. This is the online-consistent
(immediate-mode~\cite{maheswaran1999dynamic}) form of each heuristic, and using a
single engine ensures no method sees information the others lack.

The proposed auction differs in one respect, and we state plainly what it gains
from that. It batches each release-ordered window ($W{=}40$) before dispatching
it, which grants the intra-window ordering lookahead defined in
Sec.~\ref{sec:method} without breaching release-time causality. Part of its
margin over the unbatched baselines in Table~\ref{tab:main} is therefore the
horizon itself rather than the auction. Sec.~\ref{sec:decomp} separates the two
by evaluating windowed list schedulers directly. We further run three online
decomposition controls (EDF+EFT, TC-EDF+EFT, and EDF with random placement) and
a UCB1-bandit variant of the auction.

\textbf{Ceiling.} For each instance, we compute the feasibility ceiling (TC tasks
completable on some PU absent contention) and the CP-SAT offline optimum
(maximum weighted on-time throughput under \eqref{eq:ct}), reporting each
scheduler's weighted throughput as a fraction of the solver's rigorous upper
bound ($/$UB; higher is better, $1.0$ meaning the bound is met).

\textbf{Metrics and statistics.} The primary metric is $\tcr$; we also report the
best-effort completion rate $\mathrm{TCR}_{\NTC}$ and weighted throughput
($\wtc{=}2,\wntc{=}1$). Makespan is near-identical across methods (about $49$
time units) and non-discriminative, so we omit it.
We report mean $\pm$ 95\% CI and, referenced to the proposed auction, paired
Wilcoxon signed-rank tests with Holm--Bonferroni correction and Cliff's-delta
effect sizes, exploiting the fact that every scheduler runs on the identical
instances.

\textbf{Implementation.} All schedulers are in Python; PSO, GA, and DDQN
hyperparameters and convergence are in Appendix~\ref{app:hp}. LLM agents use a
lightweight model (Claude Haiku) with structured outputs, and responses are cached, and Sec.~\ref{sec:llm} reports live latency separately.
Unless noted, $W{=}40$ and $R{=}2$. We denote by \emph{CNP} the auction with a
heuristic control plane, \emph{CNP-LLM} the auction with LLM broker and edge
agents, \emph{HAS} (hierarchical agentic scheduler) the full method that adds the
LLM monitor, and \emph{CNP-bandit} the bandit variant.

\section{Results}\label{sec:results}
We report the results in the order the questions arise. Sec.~\ref{sec:overall}
establishes that the scheduler is strong and how close to optimal it is;
Sec.~\ref{sec:decomp} asks which of its components is responsible;
Sec.~\ref{sec:disruption} asks when, if ever, the LLM control plane adds
anything; and Sec.~\ref{sec:llm} measures what that plane costs to run.

\subsection{Overall Comparison}\label{sec:overall}
Table~\ref{tab:main} and Fig.~\ref{fig:ranking} report the $60$-instance
comparison. The windowed contract-net scheduler attains $\tcr=0.902\pm0.009$,
significantly above every baseline (Holm $p<0.001$, Cliff's $\delta$ from $-0.57$
to $-0.96$). The strongest baseline, GA, reaches $0.838$, and the DRL baseline
reaches $0.791$. The scheduler also records the highest weighted throughput
($274.9$), the objective it optimises. Deadline-agnostic mapping heuristics and
DRL sit at the bottom for a straightforward reason: optimising makespan or a
learned proxy does not make deadlines easier to meet.

Protecting time-critical work has a price, and the table shows it. The
best-effort rate of CNP ($\mathrm{TCR}_{\NTC}=0.731$) is the lowest in the
comparison, and its unweighted completion ($0.834$) trails EDF with random
placement ($0.849$). This is the intended consequence of time-critical-first
ordering under the $2\!:\!1$ weighting, and it is the same mechanism that lifts
weighted throughput above every deadline-agnostic method.

The optimum shows how much room is left. The scheduler operates at $0.874$ of
the rigorous CP-SAT upper bound, an optimality gap of $12.6\%$, and at about
$0.92$ of the feasibility ceiling ($0.985$). Most of the residual is the
intrinsic cost of deciding online rather than with full knowledge of the future.

\begin{table}[t]
\centering
\caption{Main comparison against all baselines over 60 instances.}

\label{tab:main}
\renewcommand{\arraystretch}{1.05}
\setlength{\tabcolsep}{2.2pt}
\footnotesize
\begin{tabular}{@{}lcccccc@{}}
\toprule
Method & $\tcr$ & $\mathrm{TCR}_{\NTC}$ & Wt.\ thr. & Wt./UB & Cliff's $\delta$ & $p_{\text{Holm}}$ \\
\midrule
\bfseries CNP & \bfseries 0.902$\pm$0.009 & \bfseries 0.731 & \bfseries 274.9 & \bfseries 0.874 & ref. & ref. \\
CNP-bandit$^{\ddagger}$ & 0.901$\pm$0.009 & 0.732 & 274.9 & 0.874 & $-0.01$ & 0.465 \\
EDF+rand$^{\dagger}$ & 0.846$\pm$0.014 & 0.854 & 271.3 & 0.863 & $-0.57$ & $<$0.001 \\
GA & 0.838$\pm$0.012 & 0.818 & 266.5 & 0.847 & $-0.68$ & $<$0.001 \\
PSO & 0.830$\pm$0.013 & 0.816 & 264.6 & 0.841 & $-0.72$ & $<$0.001 \\
LLF & 0.815$\pm$0.013 & 0.830 & 262.1 & 0.833 & $-0.81$ & $<$0.001 \\
CR & 0.813$\pm$0.014 & 0.832 & 261.6 & 0.832 & $-0.80$ & $<$0.001 \\
ERA & 0.809$\pm$0.014 & 0.790 & 257.3 & 0.818 & $-0.81$ & $<$0.001 \\
TCEDF+EFT$^{\dagger}$ & 0.809$\pm$0.014 & 0.785 & 257.0 & 0.817 & $-0.82$ & $<$0.001 \\
COVERT & 0.805$\pm$0.013 & 0.822 & 259.0 & 0.824 & $-0.85$ & $<$0.001 \\
HEFT & 0.802$\pm$0.012 & 0.819 & 258.1 & 0.821 & $-0.89$ & $<$0.001 \\
EDD / EFDF$^{\dagger}$ & 0.801$\pm$0.013 & 0.816 & 257.5 & 0.819 & $-0.88$ & $<$0.001 \\
Max-Min & 0.800$\pm$0.013 & 0.820 & 257.7 & 0.819 & $-0.89$ & $<$0.001 \\
EDF & 0.799$\pm$0.013 & 0.812 & 256.6 & 0.816 & $-0.90$ & $<$0.001 \\
Sufferage & 0.795$\pm$0.013 & 0.813 & 255.9 & 0.814 & $-0.92$ & $<$0.001 \\
FCFS & 0.795$\pm$0.013 & 0.809 & 255.5 & 0.812 & $-0.90$ & $<$0.001 \\
DDQN & 0.791$\pm$0.013 & 0.819 & 255.4 & 0.812 & $-0.91$ & $<$0.001 \\
Min-Min & 0.779$\pm$0.012 & 0.793 & 250.5 & 0.796 & $-0.96$ & $<$0.001 \\
\bottomrule
\end{tabular}
\vspace{2pt}
\begin{flushleft}\footnotesize
Mean over 60 instances; the 95\% CI is shown for the primary metric $\tcr$ only.
Wt.\ thr.\ is weighted on-time throughput ($\wtc{=}2$, $\wntc{=}1$) and Wt./UB
expresses it as a fraction of the CP-SAT upper bound (higher is better). Cliff's
$\delta$ and the Holm-corrected paired-Wilcoxon $p$ are referenced to CNP, with
negative $\delta$ indicating a method weaker than CNP. Proposed method in bold.
$^{\dagger}$EDD and EFDF coincide on these instances; EDF+rand and TC-EDF+EFT are
the \emph{online} decomposition controls. $^{\ddagger}$Bandit variant of the
proposed auction, not a baseline.
\end{flushleft}
\end{table}

\begin{figure}[t]
\centering
\includegraphics[width=\columnwidth]{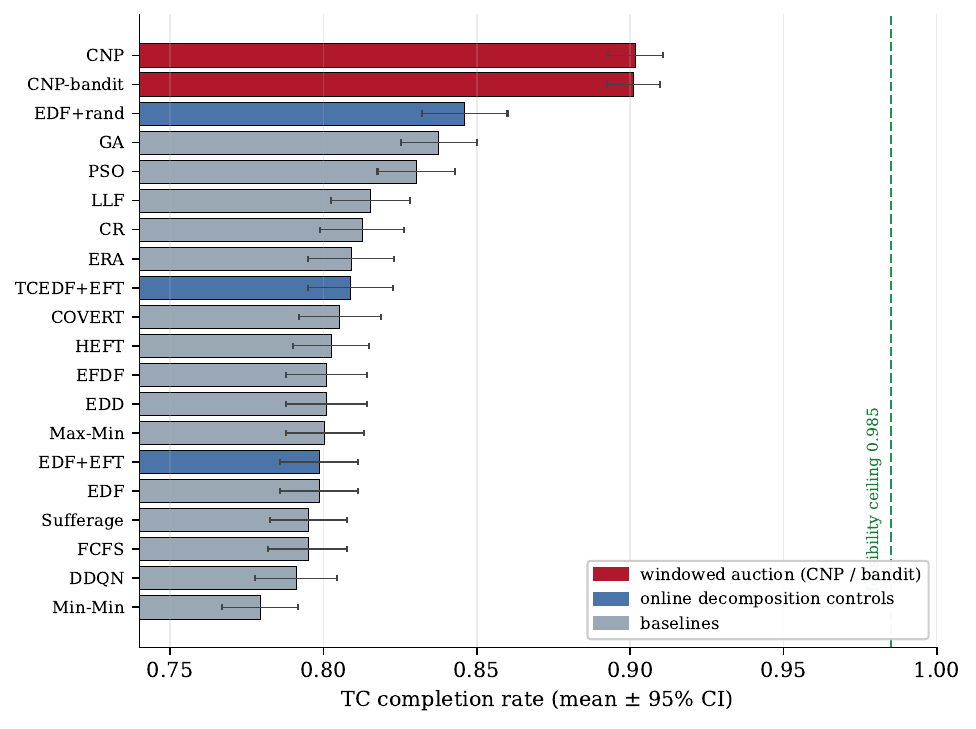}
\caption{Time-critical completion rate by method; dashed line marks the
feasibility ceiling.}
\label{fig:ranking}
\end{figure}

\subsection{Decomposition: Where the Advantage Comes From}\label{sec:decomp}
Table~\ref{tab:decomp} varies the three factors of the scheduler one at a time
over the same $60$ instances: task \emph{ordering} (plain EDF versus
time-critical-first EDF), \emph{placement} (earliest-finish-time versus a random
deadline-feasible server), and the \emph{batching horizon} $W$. Three findings
follow.

First, the batching horizon is the largest lever. Holding ordering and placement
fixed at time-critical-first and EFT, $\tcr$ climbs from $0.795$ at $W{=}1$, where
there is no lookahead at all, to $0.896$ at $W{=}40$ and $0.942$ at $W{=}200$.

Second, time-critical-first ordering supplies the other half of the gain, but it
reallocates work rather than creating it. At $W{=}40$ it raises $\tcr$ from
$0.851$ under plain EDF to $0.896$, while best-effort completion falls from
$0.845$ to $0.731$. Ordering does nothing at $W{=}1$, since a window of one task
cannot be reordered.

Third, greedy EFT placement is not the best choice. Spreading load across all
deadline-feasible servers beats it on both classes ($0.877$ against $0.851$ at
$W{=}40$ under EDF), because EFT keeps sending work to the earliest-finishing,
and usually fastest, PU and leaves slower servers idle until they are needed. The
same effect explains an ordering in Table~\ref{tab:main} that looks odd at first
glance: the online EDF+random control reaches $0.846$ and outranks online
EDF+EFT, which coincides with the plain EDF rule at $0.799$. We keep EFT in the
auction for its determinism and per-server locality, and accept the small
measured cost in $\tcr$.

Together these results place the advantage outside the auction. The windowed
TC-EDF+EFT list scheduler at $W{=}40$ reaches $0.896$ and so reproduces the full
auction ($0.902$) within their confidence intervals, which means the contract-net
machinery adds nothing over the list scheduler it implements. The UCB1 bandit
tells the same story, matching the static auction ($0.901$ against $0.902$; Holm
$p=0.465$). While the load is stationary the heuristic already runs near-optimal,
leaving no headroom for adaptation to exploit, whether by an LLM or otherwise.

\begin{table}[t]
\centering
\caption{Factorial decomposition by ordering, placement, and batching horizon.}
\label{tab:decomp}
\renewcommand{\arraystretch}{1.2}
\setlength{\tabcolsep}{4pt}
\footnotesize
\begin{tabular}{@{}llccc@{}}
\toprule
Ordering & Place & $W{=}1$ & $W{=}40$ & $W{=}200$ \\
\midrule
EDF & EFT & 0.795 (0.808) & 0.851 (0.845) & 0.869 (0.867) \\
EDF & rand & 0.836 (0.850) & 0.877 (0.877) & 0.889 (0.893) \\
TC-first & EFT & 0.795 (0.808) & \textbf{0.896 (0.731)} & 0.942 (0.630) \\
TC-first & rand & 0.836 (0.850) & 0.917 (0.762) & 0.957 (0.617) \\
\bottomrule
\end{tabular}
\end{table}

\subsection{When LLM Orchestration Helps: Non-Stationary Load}\label{sec:disruption}
Every workload so far has been stationary, and that is precisely the condition
under which adaptation has nothing to offer. We therefore introduce two
disruptions on the same instances: a mid-run surge of $40$ tightly-deadlined
safety-critical tasks, modelling a cluster of vehicles converging on a busy
intersection, and a mid-run edge-server outage.

The two differ in an important way. Under the surge the static heuristic is no
longer near-optimal: its $\tcr$ falls to about $0.76$ while the offline optimum
stays near $0.83$, leaving roughly $7$ points that a better policy could in
principle recover. Under an outage there is far less to recover, because
earliest-finish-time placement already reroutes work around the lost server.

Table~\ref{tab:disruption} and Fig.~\ref{fig:disruption} report the outcome. The
LLM control plane produces a small but \emph{statistically significant}
improvement over the static heuristic (CNP-LLM $+0.005$, $p=0.004$; HAS
$+0.006$, $p=0.005$), whereas CNP-bandit does not ($p=0.571$). Since both adapt
the same parameters, the difference points to the LLM's context-aware, zero-shot
per-window decisions rather than to adaptation in general. Under an outage, and
on stationary load (Sec.~\ref{sec:decomp}), nothing beats the static heuristic
significantly. Adding the monitor on top of the per-window agents changes little,
as HAS and CNP-LLM perform comparably.

The benefit is therefore real but conditional. It appears only when
non-stationarity opens room a better policy can recover, and even then it
captures a modest fraction of that room.

\begin{table}[t]
\centering
\caption{Completion rates under surge and outage disruptions.}

\label{tab:disruption}
\renewcommand{\arraystretch}{1.1}
\setlength{\tabcolsep}{3.5pt}
\begin{tabular}{@{}llccc@{}}
\toprule
Condition & Method & $\tcr$ & $\Delta$ & $p$ \\
\midrule
\multirow{4}{*}{Surge}
 & static (heuristic) & 0.7573$\pm$0.0125 & n/a & n/a \\
 & CNP-bandit & 0.7569$\pm$0.0125 & $-0.0004$ & 0.571 \\
 & CNP-LLM & 0.7626$\pm$0.0123 & $+0.0053$ & \textbf{0.004} \\
 & HAS & 0.7633$\pm$0.0124 & $+0.0060$ & \textbf{0.005} \\
\midrule
\multirow{4}{*}{Outage}
 & static (heuristic) & 0.8810$\pm$0.0168 & n/a & n/a \\
 & CNP-bandit & 0.8796$\pm$0.0174 & $-0.0014$ & 0.223 \\
 & CNP-LLM & 0.8824$\pm$0.0161 & $+0.0014$ & 0.313 \\
 & HAS & 0.8805$\pm$0.0167 & $-0.0005$ & 0.916 \\
\bottomrule
\end{tabular}
\end{table}

\begin{figure}[t]
\centering
\includegraphics[width=0.92\columnwidth]{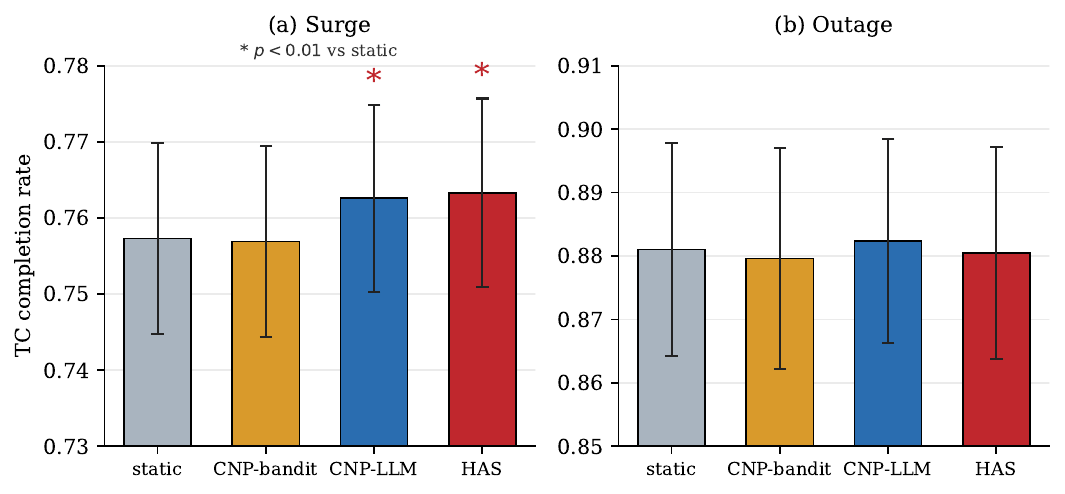}
\caption{Completion rate under surge and outage; note the truncated
$y$-axes.}
\label{fig:disruption}
\end{figure}

\subsection{LLM Control-Plane Analysis}\label{sec:llm}
A gain matters only if the layer producing it can be deployed, so we measured
what the control plane costs. On a held-out instance with live model calls it
issues $27$ schema-constrained calls at a mean latency of $3.0$\,s (median
$2.8$\,s, 95th percentile $4.3$\,s). That is orders of magnitude slower than a
data-plane window, which confirms the constraint the architecture was built
around: the LLM must stay off the per-task path and adapt on a slow cadence. It
also has a practical consequence. At realistic arrival rates a per-window call
would not return in time, so a deployment needs the stale-policy fallback of
Sec.~\ref{sec:method}, either the last committed $\Theta$ or the heuristic
default.

We also checked whether the model's explanations can be trusted. The monitor
changed its parameters at one of the two review points, and its logged rationales
are accurate in the sense that they cite the observed TC and NTC
deadline-miss rates and the parameter change follows from what they cite. In one
case it raised best-effort deferral as the NTC miss rate spiked during the surge.
Because responses are cached, each reported decision is a single sample, and
measuring how much LLM output varies across seeds and model scales remains future
work.

\subsection{Threats to Validity}
Results are from a simulator with deterministic channels and a single primary
workload size; absolute rates would differ on a physical vehicular testbed,
though the structural ordering of deadline-aware versus deadline-agnostic methods
should persist. The disruption benefit, while significant, is small and shown for
a specific surge model, and other non-stationarities may differ. The LLM findings
depend on a lightweight model; a model-scale and sampling-variance study is
future work, as are richer stochastic-channel and RSU/broker-queueing models.

\section{Conclusion}\label{sec:conclusion}
We asked when a multi-agent LLM control layer helps deadline-aware,
mixed-criticality task scheduling for autonomous vehicles at the mobile edge.

The answer has two parts. Our scheduler, a windowed contract-net auction that is
in essence a time-critical-first EDF list scheduler with earliest-finish-time
placement, outperforms $15$ standard baselines and reaches $0.87$ of a CP-SAT
upper bound. Decomposing it showed that the advantage comes from the batching
horizon and time-critical-first ordering, not from the auction and not from any
learned component. While the load stays stationary the heuristic is already close
to optimal, and neither an LLM control plane nor online adaptation improves on
it. When a surge of safety-critical tasks arrives mid-run, the picture changes:
the LLM control plane gains significantly over both that heuristic and a non-LLM
bandit, which isolates a genuine, if conditional, benefit to LLM reasoning.

For practitioners the message is simple. LLM orchestration earns its cost only
when non-stationarity opens headroom a fixed policy cannot use, and even then the
gain is modest, so the layer should be justified by the variability of the
workload rather than adopted by default. Future work will examine broader forms
of non-stationarity, larger models and their sampling variance, and a physical
vehicular testbed with stochastic channels.

\appendices
\section{Baseline Hyperparameters and Convergence}\label{app:hp}
PSO used $24$ particles, $30$ iterations, $w{=}0.7$, $c_1{=}c_2{=}1.5$. GA used a population of $30$, $30$ generations, tournament size $3$, blend crossover,
Gaussian mutation (rate $0.1$), and elitism $2$. DDQN used a two-layer MLP ($64$
hidden units), double-Q targets, a replay buffer of $5000$, a target update every
$500$ steps, $\epsilon$ annealed from $1.0$ to $0.05$, and $60$ training episodes,
with the action selecting among the three earliest-deadline ready tasks. PSO and
GA optimise per-task priorities and DDQN trains over the full instance, an offline
advantage over the strictly-online proposed method that they nonetheless fail to
convert. Fig.~\ref{fig:conv} shows convergence: PSO and GA plateau within the
allotted budget, and DDQN's episode return stabilizes.

\begin{figure}[t]
\centering
\includegraphics[width=\columnwidth]{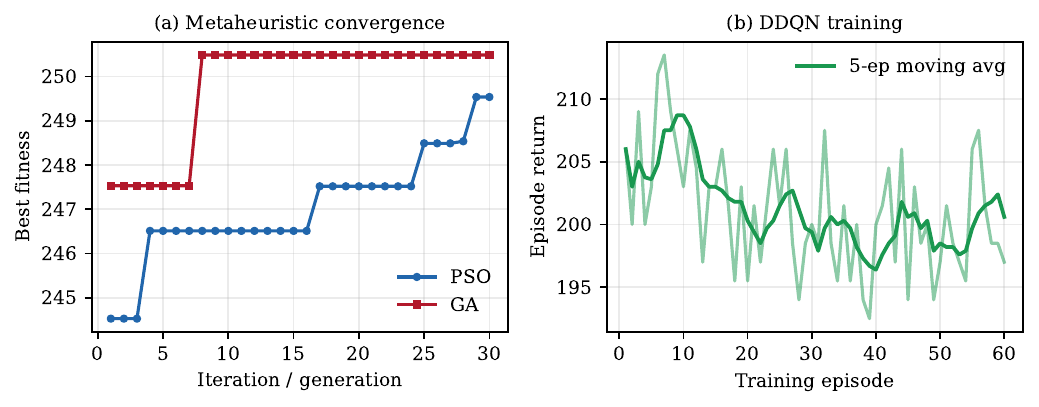}
\caption{Convergence on one instance: (a) PSO and GA best fitness; (b) DDQN raw
and 5-episode-averaged return.}

\label{fig:conv}
\end{figure}



\end{document}